\documentclass[12pt,leqno,letterpaper]{article}

\usepackage{amsmath,amsthm,enumerate,amssymb}
\usepackage[utf8]{inputenc}
\usepackage[english]{babel}

\newtheorem{theorem}{Theorem}[section]
\newtheorem{proposition}[theorem]{Proposition}
\newtheorem{lemma}[theorem]{Lemma}
\newtheorem{corollary}[theorem]{Corollary}
\newtheorem{definition}[theorem]{Definition}

\newtheorem{remark}[theorem]{Remark}
\newtheorem{example}[theorem]{Example}

\newcommand{\tr}{{\rm Tr\hskip -0.2em}~}
\DeclareMathOperator{\frechetdiff}{\mathit d}
\newcommand{\fd}[1]{\hskip-0.2em\frechetdiff\hskip -0.23em{#1}}
\newcommand{\closefd}[1]{\hskip-0.2em\frechetdiff\hskip -0.35em{#1}}
\DeclareMathOperator{\Tf}{T}

\DeclareMathOperator*{\argmin}{argmin}

\begin{document}

\title{Convex multivariate operator means}
\author{Frank Hansen}
\date{June 18, 2018\\\small Revised August 27, 2018}

\maketitle

\begin{abstract} The dominant method for defining multivariate operator means is to express them as fix-points under a contraction with respect to the Thompson metric. Although this method is powerful, it crucially depends on monotonicity. We are developing a technique to prove the existence of multivariate operator means that are not necessarily monotone. This gives rise to an entire new class of non-monotonic multivariate operator means.\\[1ex]
{\bf MSC2010:} 47A64; 53C20; 47A63\\[1ex]
{\bf{Key words and phrases:}}  Multivariate operator means; convexity, geodesically convex functions.
\end{abstract}

\section{Introduction}

The study of large classes of multivariate operator means relies on two different methods. Originally, the Karcher or least squares mean was expressed as the unique point in which a geodesically convex function attained its minimal value and there satisfies the so-called Karcher equation, or critical point equation, with respect to the logarithmic function. A number of authors studied the properties of the Karcher mean, see \cite{kn:Lim:2013,kn:Lawson:2014,kn:Lim:2012,kn:moakher:2005,kn:bhatia:2006,kn:bhatia:2012,kn:hiai:2009} and the references therein. P{\'a}lfia, Lawson and Lim 
 developed a very powerful technique that expresses the weighted multivariate mean as the unique fix-point with respect to the Thompson metric. This approach has two benefits. It makes it rather easy to establish all of the desired properties of a multivariate geometric mean as discussed by Ando, Li and Mathias \cite{kn:ando:2004:1}.  Secondly, it allows for an extension from operators on a finite dimensional Hilbert space to operators on an infinite dimensional Hilbert space. Recently, P{\'a}lfia \cite{kn:palfia:2016} extended this approach to critical point equations for arbitrary operator monotone functions $ g $ with $ g(1)=0. $ P\'alfia's method relies on the study of a specific family of convex divergence functions and very subtle inequalities between functions of operators with respect to the Thompson metric. The fix-point approach, however, has the drawback that it crucially depends on monotonicity. We are for various reasons that will be explained elsewhere interested in means that are convex but not monotone. These means have important applications even for positive numbers.

Our approach is to go back to the setting of geodesically convex functions and simplify and refine the theory to the extent that we can handle large classes of critical point equations. The drawback is that we are restricted to operators on a finite dimensional Hilbert space. The benefit is that we can not only rewrite elements of the theory of multivariate operator means in a few pages, but also extend it to non-monotonic convex means. 

The weighted geometric mean $ X\#_pY $ of two positive definite operators $ X $ and $ Y $ is defined by setting
\[
X\#_pY=Y^{1/2}\bigl(Y^{-1/2}XY^{-1/2}\bigr)^pY^{1/2}\qquad 0\le p\le 1,
\]
and it is the perspective $ \mathcal P_f(X,Y) $ of $ X $ and $ Y $ with respect to the operator monotone function $ f(t)=t^p. $ The weighted geometric mean is congruence invariant, that is the equality
\[
C^*(X\#_pY)C= C^*XC\#_pC^*YC
\]
for positive definite $ X, Y $ and arbitrary invertible $ C. $ It is also increasing in each variable.

We recall that a function $ F\colon B(\mathcal H)_+\to\mathbf R $ defined in the set $ B(\mathcal H)_+ $ of positive definite operators on a finite dimensional Hilbert space $ \mathcal H $ is said to be geodesically convex if
\[
F(X\#_pY)\le pF(X)+(1-p)F(Y)\qquad X,Y\in B(\mathcal H)_+. 
\]
We shall frequently use the transformation $ f\to\Tf(f) $ defined below.
\begin{definition}
Let $ f\colon(0,\infty)\to \mathbf R $ be a continuously differential function. We set
\[
\Tf(f)(t)=-tf'(t)
\]
for $ t>0. $
\end{definition}

\section{The Riemannian gradient}

We study operator functions defined in positive definite operators given on the form
\begin{equation}\label{The trace function F(X)}
F_A(X)=\tr f(X^{-1/2}AX^{-1/2}),
\end{equation}
where $ A $ is positive definite and $ f\colon(0,\infty)\to\mathbf R $ is a real analytic function. Since $ f $ is real analytic we may also write
\[
F_A(X)=\tr X^{-1/2}f(AX^{-1})X^{1/2}=\tr f(AX^{-1})=\tr f(A^{1/2}X^{-1}A^{1/2}).
\]
\begin{proposition}\label{calculation of Riemannian gradient}
Let $ F_A\colon B(\mathcal H)\to\mathbf R $ be a trace function on the form (\ref{The trace function F(X)}). The Riemannian gradient is then given by the operator perspective
\[
\nabla_X F_A(X)=\mathcal P_g(A,X)=X^{1/2}g\bigl(X^{-1/2}AX^{-1/2}\bigr)X^{1/2}
\]
of the function $ g=\Tf(f). $ 
\end{proposition}

\begin{proof}
We write $ F_A(X)=\tr f(AX^{-1}) $ and calculate the Fréchet differential $ \fd{}F_A(X) $ in a self-adjoint operator $ V $ and obtain
\[
\begin{array}{rl}
\fd{}F_A(X)V&=\tr\closefd{}f\bigl(AX^{-1}\bigr)\fd(AX^{-1})V\\[2ex]
&=\tr f'\bigl(AX^{-1}\bigr)A\,\fd(X^{-1})V\\[2ex]
&=-\tr f'\bigl(AX^{-1}\bigr)A X^{-1}VX^{-1}\\[2ex]
&=-\tr X^{1/2} f'\bigl(X^{-1/2}AX^{-1/2}\bigr)X^{-1/2}AX^{-1}VX^{-1}\\[2ex]
&=-\tr X^{-1/2} f'\bigl(X^{-1/2}AX^{-1/2}\bigr)X^{-1/2}AX^{-1}V,
\end{array}
\]
where we used the formula $ \tr\closefd{}f(x)h=\tr f'(x)h, $ cf.~\cite[Theorem 2.2]{kn:hansen:1995}.
The Riemannian gradient $ \nabla_X F_A(X) $ is defined by the relation
\[
<\nabla_X F_A(X)\mid V>_X=\fd{}F_A(X)V
\]
for arbitrary $ V. $ That is
\[
\tr X^{-1}\nabla_X F_A(X) X^{-1}V=-\tr X^{-1/2} f'\bigl(X^{-1/2}AX^{-1/2}\bigr)X^{-1/2}AX^{-1}V,
\]
and since $ V $ is arbitrary
\[
X^{-1}\nabla_X F_A(X) X^{-1}=-X^{-1/2} f'\bigl(X^{-1/2}AX^{-1/2}\bigr)X^{-1/2}AX^{-1}.
\]
By multiplying from the left and the right with $ X $ we obtain
\[
\begin{array}{rl}
\nabla_X F_A(X) &= -X^{1/2} f'\bigl(X^{-1/2}AX^{-1/2}\bigr)X^{-1/2}A\\[2ex]
&= -X^{1/2} f'\bigl(X^{-1/2}AX^{-1/2}\bigr)X^{-1/2}AX^{-1/2}X^{1/2}\\[2ex]
&=X^{1/2} g(X^{-1/2}AX^{-1/2})X^{1/2}\\[2ex]
&=\mathcal P_g(A,X),
\end{array}
\]
which is the statement to be proved.
\end{proof}

If $ g(1)=0, $ it follows that the Riemannian gradient of the above function $ F_A(X) $ vanishes for $ X=A. $
Note also that the map $ f\to\Tf(f) $ is linear. The following lemma is elementary.

\begin{lemma}\label{f versus tilde f}
We list $ g=\Tf(f) $ for some functions $ f. $
\[
\begin{array}{ll}
f(t)=\log(t+\lambda) &g(t)=\displaystyle\frac{-t}{t+\lambda}\\[2ex]
f(t)=\displaystyle\frac{(\log t)^2}{2}& g(t)=-\log t\\[2ex]
f(t)=t\log t-t\qquad& g(t)=-t\log t\\[2ex]
f(t)=t^p &g(t)=-p t^p,
\end{array}
\]
where $ \lambda $ and $ p $ are constants, and $ \lambda\ge 0. $
\end{lemma}

\section{A class of geodesically convex functions}

\begin{definition}
We call a function $ f\colon(0,\infty)\to\mathbf R $ convex-log, if it can be written on the form
\[
f(t)=\Phi(\log t)\qquad t>0,
\]
where $ \Phi\colon \mathbf R\to\mathbf R $ is a convex function. We say that $ f $ is strictly convex-log, if $ \Phi $ is strictly convex.
\end{definition}

A convex-log function satisfies the inequality
\begin{equation}\label{inequality for convex-log functions}
f(t^ps^{1-p})\le pf(t)+(1-p)f(s)\qquad t,s>0
\end{equation}
for $ p\in[0,1]. $ Indeed,
\[
\begin{array}{l}
f(t^ps^{1-p})=\Phi\bigl(\log(t^p s^{1-p})\bigr)=\Phi(p\log t+(1-p)\log s)\\[1.5ex]
\le p\Phi(\log t)+(1-p)\Phi(\log s)=p f(t)+(1-p) f(s).
\end{array}
\]
\begin{definition}
Let $ X $ and $ Y $ be self-adjoint $ n\times n $ matrices. We denote by $ x_{[1]},\dots,x_{[n]} $ and $ y_{[1]},\dots,y_{[n]} $ the eigenvalues of respectively $ X $ and $ Y $ counted with multiplicity and decreasingly ordered. We say that $ X $ is majorized by $ Y, $ and we write $ X\prec Y, $ if
\[
\sum_{i=1}^k x_{[i]}\le\sum_{i=1}^k y_{[i]}\qquad k=1,\dots,n-1
\]
and $ \tr X=\tr Y. $
\end{definition}

Ando and Hiai \cite[Corollary 2.3]{kn:ando:1994:2} proved the following result.

\begin{theorem}\label{Ando-Hiai theorem}
Let $ X $ and $ Y $ be positive definite operators acting on a Hilbert space of finite dimension. Then
\[
\log(X\#_pY)\prec p\log X + (1-p)\log Y
\]
for $ p\in[0,1]. $
\end{theorem}
Chi-Kwong Li and Yiu Poon \cite[Theorem 3.3]{kn:li:2011} proved the following result.

\begin{theorem}\label{existence of unital trace preserving map}
Let $ A $ and $ B $ be self-adjoint operators acting on a Hilbert space $ \mathcal H $ of finite dimension. There exists a unital trace preserving completely positive linear map $ T\colon B(\mathcal H)\to B(\mathcal H) $ such that
\[
T(A)=B,
\]
if and only if  $ B\prec A. $ 
\end{theorem}

\begin{theorem}
 let $ f $ be a convex-log function. The trace function
\[
G(X)=\tr f\bigl(X\bigr)
\]
is geodesically convex in positive definite operators acting on a Hilbert space $ \mathcal H $ of finite dimension. If $ f $ is strictly convex-log then $ G(X) $ is strictly geodesically convex.
\end{theorem}

\begin{proof}
By definition there exists a convex function $ \Phi\colon\mathbf R\to\mathbf R $ such that
\[
G\bigl(X\#_pY\bigr)=\tr f\bigl(X\#_pY\bigr)=\tr\Phi\bigl(\log(X\#_pY)\bigr)
\]
for positive definite $ X $ and $ Y, $ and $ p\in[0,1]. $ Since by Theorem~\ref{Ando-Hiai theorem}
\[
\log(X\#_pY)\prec p\log X + (1-p)\log Y,
\]
there exists according to Theorem~\ref{existence of unital trace preserving map} a unital trace preserving linear completely positive map $ T\colon B(\mathcal H)\to B(\mathcal H) $ such that
\[
T\bigl(p\log X + (1-p)\log Y\bigr)=\log(X\#_pY).
\]
This implies the existence of matrices $ a_1,\dots,a_n $ with
\[
a_1^*a_1+\cdots+a_n^*a_n= I\qquad\text{and}\qquad a_1a_1^*+\cdots+a_na_n^*= I
\]
such that
\[
\log\bigl(X\#_pY\bigr)=\sum_{i=1}^n a_i^*\bigl(p\log X+(1-p)\log Y\bigr)a_i.
\]
We thus obtain
\[
\begin{array}{l}
\displaystyle G\bigl(X\#_pY\bigr)=\tr\Phi\left(\sum_{i=1}^n a_i^*\bigl(p\log X+(1-p)\log Y\bigr)a_i\right)\\[2ex]
\displaystyle\le\tr\sum_{i=1}^n a_i^*\Phi\bigl(p\log X+(1-p)\log Y\bigr)a_i\\[3.5ex]
=\tr \Phi\bigl(p\log X+(1-p)\log Y\bigr)\le p\tr\Phi(\log X)+(1-p)\tr\Phi(\log Y)\\[2ex]
=pG(X)+(1-p)G(Y),
\end{array}
\]
where we used Jensen's trace inequality \cite[Theorem 2.4]{kn:hansen:2003:2} and the convexity of $ \Phi $ under the trace.  We proved that $ G(X) $ is geodesically convex, and
we also realize that $ G(X) $ is strictly geodesically convex if $ \Phi $ is strictly convex.
\end{proof}

\begin{corollary}\label{main geodesically convex function}
Let $ A $ be a positive definite operator, and let $ f $ be a real analytic strictly convex-log function. The trace function
\[
F_A(X)=\tr f\bigl(X^{-1/2}AX^{-1/2}\bigr)
\]
is strictly geodesically convex in positive definite operators $ X. $ 
\end{corollary}

\begin{proof} Let $ G(X)=\tr f(X) $ be the strictly geodesically function studied in the preceding theorem. Since $ f $ is real analytic we obtain
\[
F_A(X)=\tr f\bigl(A^{1/2}X^{-1}A^{1/2}\bigr)=G\bigl(A^{1/2}X^{-1}A^{1/2}\bigr)
\]
and thus
\[
\begin{array}{rl}
F_A(X\#_pY)&=G\bigl(A^{1/2}(X\#_pY)^{-1}A^{1/2}\bigr)=G\bigl(A^{1/2}(X^{-1}\#_pY^{-1})A^{1/2}\bigr)\\[1.5ex]
&=G\bigl((A^{1/2}X^{-1}A^{1/2})\#_p(A^{1/2}Y^{-1}A^{1/2})\bigr)\\[1.5ex]
&\le p G(A^{1/2}X^{-1}A^{1/2})+(1-p)G(A^{1/2}Y^{-1}A^{1/2})\\[1.5ex]
&=p F_A(X)+(1-p)F_X(Y)
\end{array}
\]
for positive definite $ X $ and $ Y, $ and $ p\in[0,1]. $ We realize that $ F_A(X) $ is strictly geodesically convex.
\end{proof}

\begin{remark}
Note that if $ f(t)=\Phi(\log t) $ is convex-log, then 
\[
\Tf(f)(t)=-t f'(t)=-t\Phi'(\log t)t^{-1}=-\Phi'(\log t)
\]
and thus $ \Tf(f)'(t)=-\Phi''(t) t^{-1}\le 0, $ so $ g=\Tf(f) $ is necessarily decreasing.
\end{remark}

\begin{proposition}\label{proposition about T(f)}
Let $ g\colon(0,\infty)\to\mathbf R $ be a real analytic strictly decreasing function. There exists a strictly convex-log function $ f $ such that $ g=\Tf(f). $ In addition, if $ g(1)=0 $ then $ f'(1)=0. $\end{proposition}

\begin{proof}
We define $ \Phi\colon\mathbf R\to\mathbf R $ up to an additive constant by setting
\[
\Phi'(x)=-g(e^x)\qquad x\in\mathbf R.
\]
Then $ \Phi $ is strictly convex, indeed $ \Phi''(x)=-g'(e^x)e^x>0 $ for $ x\in\mathbf R. $ The function
\[
f(t)=\Phi(\log t)\qquad t>0
\]
is thus strictly convex-log. In addition,
\[
\Tf(f)(t)=-t f'(t)=-t \Phi'(\log t)t^{-1}=g\bigl(e^{\log t}\bigr)=g(t).
\]
It also follows that $ f'(1)=0 $ if $ g(1)=0. $
\end{proof} 

We thus realize that the operator perspective $ \mathcal P_g(A,X) $ of a real analytic strictly decreasing function $ g $ may be obtained as the Riemannian gradient of the strictly geodesically convex function
\[
F_A(X)=\tr f\bigl(X^{-1/2}AX^{-1/2}\bigr),
\]
where $ f $ is chosen according to the above proposition. In \cite[Introduction]{kn:palfia:2016} another approach is given in terms of divergence functions in the special case, where $ -g $ is operator monotone with $ g(1)=0. $

\begin{example}
By different choices of convex functions $ \Phi $ we exhibit various geodesically convex functions of the form
$ F_A(X)=\tr f(X^{-1/2}AX^{-1/2}). $
\[
\begin{array}{lll}
\Phi(x)=\frac{1}{2}x^2 \qquad &f(t)=\frac{1}{2}(\log t)^2 \qquad &g(t)=-\log t\\[2ex]
\Phi(x)=\log(e^x+\lambda) & f(t)=\log(t+\lambda) &g(t)=-\displaystyle\frac{t}{t+\lambda}\\[2ex]
\Phi(x)=\exp(px) & f(t)=t^p & g(t)=-p t^p,
\end{array}
\]
where $ g=\Tf(f), $ $ \lambda\ge 0 $ and $ p\in\mathbf R. $
\end{example}

\section{Multivariate operator means}

Let $ g\colon(0,\infty)\to\mathbf R $ be a real analytic strictly decreasing function with $ g(1)=0, $ and choose $ f $ such that $ g=\Tf(f), $ cf. Proposition~\ref{proposition about T(f)}.
Let $ \omega=(\omega_1,\dots,\omega_k) $ be a probability vector and take a $ k $-tuple $ \mathcal A=(A_1,\dots,A_k) $ of positive definite operators. Since $ f $ is convex-log the trace function
\begin{equation}\label{convex combination of geodesically convex functions}
F_{\mathcal A}(X)=\sum_{i=1}^k \omega_i F_{A_i}(X)=\sum_{i=1}^n \omega_i\tr f\bigl(X^{-1/2}A_iX^{-1/2}\bigr)
\end{equation}
is strictly geodesically convex in positive definite operators $ X, $ cf. Corollary~\ref{main geodesically convex function}. If $ F_{\mathcal A}(X) $ is bounded from below then it admits a unique minimizer, and we then define
\[
X_g(\omega; A_1,\dots,A_k)=\argmin_{\substack{0<X}}F_{\mathcal A}(X)
\]
as the uniquely defined positive definite operator $ X $  in which $ F_{\mathcal A}(X) $ attains its minimal value. Furthermore, the minimizer $ X=X_g(\omega; A_1,\dots,A_k) $ is also the uniquely defined positive definite operator in which the Riemannian gradient vanishes, that is where
\begin{equation}\label{critical point equation}
\nabla_X F_{\mathcal A}(X)=\sum_{i=1}^k \omega_i \mathcal P_g(A_i,X)=0.
\end{equation}
We call $ X_g(\omega; A_1,\dots,A_k) $ the weighted operator mean associated with the function $ g $ and the vector of weights $ \omega. $ Since $ g(1)=0 $ we realize that
\[
X_g(\omega; A,\dots,A)=A.
\]
We also note that a multivariate mean  $ X_g(\omega;A_1,\dots,A_k) $ is symmetric in its $ k $ arguments if $ \omega=(k^{-1},\dots,k^{-1}). $ In this case we drop the reference to the probability vector $ \omega $ and plainly write $ X_g(A_1,\dots,A_k). $

\begin{theorem}\label{main theorem}
Let $ g\colon(0,\infty)\to\mathbf R $ be a strictly decreasing real analytic function with $ g(1)=0 $ that is  either operator convex or operator concave. The weighted operator mean $ X_g(\omega; \mathcal A)=X_g(\omega; A_1,\dots,A_k) $  is congruence invariant, that is the identity
\begin{equation}\label{condition for congruence invariance} 
X_g\bigl(\omega; C^*A_1C,\dots,C^*A_kC)=C^*X_g(\omega; A_1,\dots,A_k)C
\end{equation}
holds for every invertible operator $ C $ on the underlying Hilbert space. If $ g $ is operator convex then
\begin{equation}\label{concave mean}
X_g(\omega; A_1,\dots,A_k)\le\omega_1 A_1 +\cdots+\omega_k A_k
\end{equation}
and if $ g $ is operator concave then
\begin{equation}\label{hyper-mean}
X_g(\omega; A_1,\dots,A_k)\ge\omega_1 A_1 +\cdots+\omega_k A_k.
\end{equation}
\end{theorem}

\begin{proof}
If $ g $ is operator convex or operator concave, then the perspectives $ \mathcal P_g(A_i,X) $ are congruence invariant for $ i=1,\dots,k $ by \cite[Proposition 2.3]{kn:hansen:2014:4}. For an invertible operator $ C $ we thus obtain
\[
\sum_{i=1}^k \omega_i \mathcal P_g(C^*A_iC,C^*XC)=C^*\sum_{i=1}^k \omega_i \mathcal P_g(A_i,X)C=0
\]
for $ X=X_g(\omega; A_1,\dots,A_k). $ The uniqueness of the solution to the critical point equation~(\ref{critical point equation}) then implies that
\[
C^*X_g(\omega; A_1,\dots,A_k)C=X_g(\omega; C^*A_1C,\dots,C^*A_kC),
\]
which is the first statement. If $ g $ is operator convex, we obtain by setting $ Y=X_g(\omega; A_1,\dots,A_k) $ the inequality
\[
g\left(\sum_{i=1}^k \omega_i Y^{-1/2}A_iY^{-1/2}\right)\le\sum_{i=1}^k \omega_ig\bigl(Y^{-1/2}A_iY^{-1/2}\bigr)=0,
\]
where we used that the last sum is similar to the vanishing Riemannian gradient. Since $ g $ is decreasing with $ g(1)=0 $ we derive that
\[
\sum_{i=1}^k \omega_i Y^{-1/2}A_iY^{-1/2}\ge I
\]
and the first statement follows. The case when $ g $ is operator concave is similarly proved.
\end{proof}

If in the above theorem $ g\colon(0,\infty)\to\mathbf R $ is assumed to be operator convex, then $ -g $ is operator monotone (since it is also numerically monotone), and we realise that the first part of Theorem~\ref{main theorem} is already contained in the theory of P{\'a}lfia \cite{kn:palfia:2016}. We therefore know that the multivariate mean $  X_g(\omega;A_1,\dots,A_k) $ in this case is operator concave and increasing.

If however $ g $ is operator concave then we are in new territory. Since in this case $  X_g(\omega;A_1,\dots,A_k) $ satisfies inequality (\ref{hyper-mean}) we call such means for hyper-means. Based on various examples we suspect that hyper-means are convex; but they are in generel not monotonic.

\subsection{Bivariate operator means}

If a multivariate operator mean $ X_g(\omega;A_1,\dots,A_k) $ is generated by an either operator convex or operator concave function $ g\colon(0,\infty)\to\mathbf R, $ then it is congruence invariant and therefore the perspective of its restriction to $ k-1 $ variables, cf. also \cite[Proposition 3.3]{kn:hansen:2014:4}. This implies for $ k=2 $ that 
\[
X_g(\omega;A_1,A_2)=A_2^{1/2}X_g\bigl(\omega; A_2^{-1/2}A_1A_2^{-1/2},I\bigr)A_2^{1/2} 
\]
 is the perspective of the function
\[
\varphi(t)=X_g(\omega;t,1)\qquad t>0.
\]
This function satisfies the critical point equation (\ref{critical point equation}) taking the form
\[
\nabla_x F_{(t,1)}(x)=\omega_1\mathcal P_g(t,x)+\omega_2\mathcal P_g(1,x)=0
\]
for $ x=\varphi(t). $ By dividing with $ x $ this gives the equation
\begin{equation}\label{equation to determine the representing function}
\omega_1g(x^{-1}t)+\omega_2g(x^{-1})=0,
\end{equation}
which may be used to determine the function $ \varphi(t)=x. $
We call $ g $ the generating function for the mean $ X_g(\omega;\mathcal A), $ and we call  $ \varphi $ the representing function for the bivariate mean  $ X_g(\omega;A_1,A_2).  $

\section{Examples of multivariate operator means}

There are unfortunately no general formulas for multivariate operator means satisfying a critical point equation as in (\ref{critical point equation}) except in the case $ k=2. $ We shall give two examples of multivariate operator means and calculate their representing functions in the case $ k=2. $

\begin{example}
We examine the function
\[
F_{\mathcal A}(X)=\frac{1}{k}\sum_{i=1}^k \tr f\bigl(X^{-1/2}A_iX^{-1/2}\bigr)
\]
for positive definite operators $ \mathcal A=(A_1,\dots,A_k), $ where
\[
f(t)=t^{-p}+p t\qquad t>0
\]
for $ 0<p<1. $ Note that $ f $ and thus $ F_{\mathcal A}(X) $ are bounded from below by zero. Since $ f'(t)=p(t^{-(p+1)}-1) $ we obtain
\[
g(t)=\Tf(f)(t)=-tf'(t)=p(t^{-p}-t)\qquad t>0.
\]
Note that $ g $ is strictly decreasing with $ g(1)=0 $ and operator convex, hence $ -g $ is operator monotone.  The representing function $ \varphi(t)=x $ satisfies
\[
{\textstyle\frac{1}{2}}g(x^{-1}t)+{\textstyle\frac{1}{2}}g(x^{-1})=0
\]
according to (\ref{equation to determine the representing function}),
which after inserting $ g(t) $ takes the form
\[
p \bigl((x^{-1}t)^{-p}-x^{-1}t\bigr)+p\bigl( (x^{-1})^{-p}-x^{-1}\bigr)=0,
\]
or after dividing with $ p $ and rearranging the terms
\[
(t^{-p}+1)x^p=(t+1)x^{-1}
\]
with solution
\[
\varphi(t)=x=\left(\frac{t+1}{t^{-p}+1}\right)^{\frac{1}{p+1}}\qquad t>0.
\]
Since $ g $ is operator convex, it follows from P{\'a}lfia's theory that $ \varphi(t) $ is operator concave, hence operator monotone. However, $ \varphi(t) $ is not one of the known operator monotone functions. 
\end{example}

\begin{example}
Consider the function
\[
g(t)=p(1-t^p)\qquad t>0
\]
for $ 1\le p\le 2. $ 
It is numerically decreasing with $ g(1)=0 $ and operator concave. The associated convex function $ \Phi $ satisfying
\[
\Phi'(x)=-g(e^x)=p(e^{px}-1)\qquad x\in\mathbf R
\]
is up to a constant given by
\[
\Phi(x)=e^{px}-px.
\]
Note that $ \Phi $ is not operator convex, and that
\[
f(t)=\Phi(\log t)=t^p-p\log t\qquad t>0
\]
is convex-log and bounded from below by the constant $ 1. $ The representing function $ \varphi(t) $ for the associated symmetric convex operator mean of two variables $ X=X_g(A,B) $ is  given, cf. equation (\ref{equation to determine the representing function}), as the unique solution $ x=\varphi(t) $ to the functional equation
\[
g(x^{-1}t)+g(x^{-1})=0.
\]
Inserting $ g $ we obtain
\[
p\bigl((x^{-1}t)^p-1\bigr)+p\bigl(x^{-1}-1\bigr)=0
\]
with solution
\[
\varphi(t)=x=\left(\frac{t^p+1}{2}\right)^{1/p}\qquad t>0.
\]
Since $ g $ is operator concave, we have no general theory to assert that $ \varphi(t) $ is operator convex. However, the statement follows from \cite[Theorem 2.2]{kn:hansen:2014}. We thus obtain that the associated operator mean of two variables $ X_g(A,B) $ is convex and has a unique extension to a multivariate operator mean satisfying the critical point equation in (\ref{critical point equation}).
\end{example}

{\bf Acknowledgements.}  
The author acknowledges support from the Japanese government Grant-in-Aid for scientific research 26400104.

{\small


}

\end{document}